\documentclass[conference]{IEEEtran}

\usepackage{url}
\usepackage{ifthen}
\usepackage{cite}
\usepackage[final]{graphicx}

\usepackage{amsthm}
\usepackage[cmex10]{amsmath}
\usepackage{amssymb}
\usepackage{euscript}
\usepackage{subfigure}
\usepackage[colorinlistoftodos]{todonotes}
\usepackage[a-1b]{pdfx}
\usepackage{hyperref}
\usepackage{bbm}
\usepackage{mathtools}

\newtheorem{theorem}{Theorem}

\newtheorem{lemma}{Lemma}

\newcommand{\nn}{\nonumber\\}

\begin{document}
\title{Data Freshness in Leader-Based Replicated Storage} 

\author{\IEEEauthorblockN{Amir Behrouzi-Far, Emina Soljanin and Roy D. Yates}
\IEEEauthorblockA{\textit{Department of Electrical and Computer Engineering, Rutgers University} \\
\{amir.behrouzifar,emina.soljanin,ryates\}@rutgers.edu}}

\maketitle
\begin{abstract}
Leader-based data replication improves consistency in highly available distributed storage systems via sequential writes to the ``leader'' nodes. After a write has been committed by the leaders, ``follower'' nodes are written by a multicast mechanism and are only guaranteed to be eventually consistent. With Age of Information (AoI) as the freshness metric, we characterize how the number of leaders affects the freshness of the data retrieved by an instantaneous read query. In particular, we derive the average age of a read query for a deterministic model for the leader writing time and a probabilistic model for the follower writing time.  We obtain a closed-form expression for the average age for exponentially distributed follower writing time. Our numerical results show that, depending on the relative speed of the write operation to the two groups of nodes, there exists an optimal number of leaders which minimizes the average age of the retrieved data, and that this number increases as the relative speed of writing on leaders increases. 
\end{abstract}

\section{introduction}
Databases are a central part of online services \cite{chang2008bigtable,decandia2007dynamo}. They constantly execute write and read operations from different clients/services. A large fraction of write operations are updates to the currently stored data \cite{armstrong2013linkbench}. Each update increments the version of the data and sets its timestamp to the current time. Read operations, on the other hand, ask for the latest version of the data, which also has the most recent timestamp. A write/read operation is consistent if it writes to or reads from the latest version of the data.

Data availability is a major objective in database systems \cite{lakshman2010cassandra,ford2010availability}. To achieve high availability, data is partitioned into multiple chunks, and each chunk gets replicated in multiple storage nodes \cite{wiesmann2000understanding,aktas2019load}. These replicas can possibly be located in different geographic locations to further improve data availability and storage robustness. However, according to the CAP theorem \cite{gilbert2002brewer}, availability and consistency are not achievable at the same time in a distributed system. Therefore, availability-consistency trade-offs are inevitable in a distributed database.

Leader-based (replicated) databases have been proposed to increase data consistency in large scale distributed storage \cite{corbett2013spanner,sharov2015take}. In these databases, some of the storage nodes assigned to a given chunk of data are elected as leaders. Every write operation to the data chunk is first written to the leaders, where sequential writes from one leader to the next guarantee that they always have the latest version of the data. Thus, write consistency is guaranteed at the leaders. After sequential writes to the leaders, the update is committed and gets transmitted to a set of follower nodes via a multicast mechanism. After an update is committed, it becomes available to read and further updates may be initiated. Due to the guaranteed leaders' consistency and high availability (by the existence of the replicas), leader-based systems are widely deployed in practice. For example, Google Spanner \cite{corbett2013spanner}, Amazon DynamoDB \cite{mathew2014overview}, Apple FoundationDB \cite{chrysafis2019foundationdb} have used it as a part of their storage ecosystem, where they employ the Paxos algorithm \cite{chandra2007paxos} for leader election among the storage nodes. The Raft algorithm \cite{ongaro2014search}, an alternative for Paxos, has been used by MongoDB and InfluxDB as a leader election algorithm. These databases are $eventually$ consistent in that the data is guaranteed to be consistent among the leader nodes, and the data on each follower eventually will become consistent, in the absence of subsequent writes.

Freshness of retrieved data in eventually consistent database systems have been studied in \cite{rahman2017characterizing,golab2011analyzing,bailis2014quantifying,bannouradaptive}. With a probabilistic definition of consistency, \cite{rahman2017characterizing} studies the delay-consistency trade-off in a distributed database system. Authors in \cite{bailis2014quantifying} establish a version-based staleness metric and investigate the delay benefits of eventually consistent databases and shed light on why they are good enough in practice. Data staleness in dynamo-style quorum-based replicated storage systems is studied in \cite{zhong2018minimizing}. In this work, using AoI as the freshness metric, the average age of the retrieved data is studied for different types of quorum consensus. AoI has been recently used  to characterize data freshness in various types of data transmission systems \cite{bastopcu2019age,Arafa-YUP-arxiv2018,Bacinoglu-SUBM-isit2018,Kadota-SUBSM-arxiv2018,Beytur-UB-SIU2018,Sert-SBUB-SIU2018,Arafa-Ulukus-asilomar2017,Baknina-Ulukus-arxiv2018coded,bastopcu2019age2,Kadota-UBSM-allerton2016,He-YE-wiopt2016,Zhong-YS-allerton2017,Zhong-YS-isit2018}.

Despite the existing work on eventually consistent databases, we are far from understanding the dynamics governing data freshness in such systems. In particular, in a leader-based model, we know very little about the effect of the number of leaders on the freshness of the data retrieved by a read operation. Having more leaders improves the consistency of the database since leaders are guaranteed to have the latest version of the data. On the other hand, due to the sequential, and thus time-costly writes, having many leaders can prolong the commit time of an update, which makes the update more stale when it becomes available to read. Furthermore, the relative speed of writing an update to the leaders, compared to the multicast writing to the follower nodes, is a critical factor in optimizing system parameters for increasing the freshness of data retrieved by a read operation.

Using the AoI metric\cite{kaul2012real,Kaul-Yates-isit2018priority,Kaul-YG-infocom2012,SunCyr-spawc2018,Kosta-PEA-isit2017,Yates-NSZ-isit2017,Zhong-YS-isit2017,Kaul-Yates-isit2017,Yates-arxiv2018networks,Sun-UBYKS-IT2017UpdateorWait,Zhong-YS-spawc2018}, we study the average freshness of the data retrieved by a read query in a leader-based replicated database. We develop two different models for write operations to 1) the leaders, which are written by sequential writes, and 2) the followers, which are written through a multicast. We then derive the dependence of the average age of the data retrieved by a read query on our model parameters. We assume the writing time to the leaders is deterministic and scales linearly with the number of leaders. However, it is a random variable for each follower. For the exponential write time distribution of the follower nodes, we drive a closed form for the average age of the retrieved data by a read operation. Our numerical results show that, depending on the relative speed of the write operation to the two groups of nodes, there exists an optimal number of leaders which minimizes the average age of the retrieved data. In addition, we observe that the average age monotonically increases or decreases with the number of leaders, depending the relative speed of the write operations. Furthermore, for a system with time-varying demands, we show that it is possible to increase the number of followers without increasing the average age, when the data demand gets high.

\section{System model and Problem Formulation}
Consider a leader-based database system, where each data chunk is replicated on $n$ nodes. Out of the $n$ nodes storing the same chunk one is elected as the \textit{leader}. To ensure write consistency, any write query first has to update the leader. To further improve consistency, the leader may initiate a series of sequential writes some other nodes, which with the original leader form the \textit{leaders} group. We refer to the rest of the nodes as \textit{followers}. An update is \textit{initiated} to the leaders and is \textit{committed} once it is written to all leaders. Each update is written in an update \textit{round}, which increments the version of the data and updates its timestamp to the update initiation time.  A node is said to be consistent if it is written with the latest update. Thus all leaders get consistent in every update round, the followers may only be eventually consistent.

In our model, fresh updates are initiated to the leaders at time instances $0,c,2c,\dots$. 
The write operation requires time $c$ to complete, and thus the $k$th update initiated at time $kc$ is committed at time $(k+1)c$. We refer to the time interval of length $c$ as a write \textit{slot}. After a write is committed to all leaders, it is forwarded to the followers, and the next write is initiated by the system. An update becomes available to read only after it is committed. We assume that forwarding an update to the followers happens through a multicast mechanism so that a random number of follower nodes also receive the latest update in a given round. Specifically, we assume the write time to each follower node is described by the random variable $T_w$. Thus a follower receives the latest update only if $T_w<c$.  We also assume that the $T_w$ are iid across the followers, with CDF $F_w(t)$. We refer to a follower with the latest version of the data as \textit{consistent}; otherwise a followers is \textit{non-consistent}. Thus, the set of consistent followers may be different in each round.

When a read query is generated, it is sent to $r$ nodes. The \textit{query size} parameter $r$ is fixed across the queries. After reading from $r$ possibly different versions of data, the system responds to the query with the most recent version among the $r$ accessed nodes. A read query is consistent if it retrieves the latest version of the data. The consistency of a read query is guaranteed if it is sent to at least one leader. However, sending every read query to a leader may overload the nodes in the leader group, which in turn reduces their availability \cite{vogels2009eventually}. In practical, implementations of a leader-based database, e.g.\ Amazon DynamoDB, consistency of the read queries is traded for availability of the leader nodes by sending each query to $r$ \textit{randomly} chosen nodes.

In our model, the system cancels the write of the $k$th update on a follower if it is not written by the time $(k+2)c$, i.e. the end of time slot $k+1$. In other words, the $k$th update gets preempted when update $k+1$ is initiated to the followers. Thus, the age process at the followers is statistically the same across both the followers and the time slots. The age process at the leaders is the same across the leaders and the time slots, as well. Therefore, to study the age of a read query it is sufficient to study the age process at the nodes in a time slot $[kc,(k+1)c)$, for any $k\in\{0,1,2,\dots\}$. We assume the system has started running at $-\infty$ and the arrival time of a read query $T_a$ is a uniform random variable in the range $[0,c)$. Arrivals in time slots $[kc,(k+1)c), k=1,2,\dots$ observe statistically the same age instances at the nodes, as arrivals in the interval $[0,c)$. We further note that, practical evaluations in \cite{bailis2012probabilistically} show that the average latency of read operations is an order of magnitude smaller than the write operations. Therefore, we assume the read queries are \textit{instantaneous} in this work.

We study the average \textit{age} of the data retrieved by a read query, defined as follows. The age of information $\Delta_i(t)$ (or simply the \textit{age}) at node $i$ at time $t$ is the time difference between time $t$ and the timestamp $u_i(t)$ of the latest data update at node $i$:
    \begin{equation}
        \Delta_i(t)=t-u_i(t).
    \end{equation}
The age of a read query $\Delta(t)$ at time $t$ is defined as the minimum age over the set $\CMcal{S}_r$ of nodes the query is sent to:
    \begin{equation}
        \Delta(t)=\underset {i\in \CMcal{S}_r}{\textup{min}} \Delta_i(t).
    \end{equation}

Leaders experience a different age process than the followers since, in each update round, leaders are guaranteed to get the most recent update of the data while the followers may or may not receive it. The age processes at a leader and a follower are shown in Fig.~\ref{fig:ageEvol}. Leaders are written in every update round and their data becomes available only after every leader is updated and the write is committed, which is $c$ units of time after the update initiation. Thus, the age process is the same for every leader node, as shown by $\Delta_{\text{leader}}$ in Fig. \ref{fig:ageEvol}. 

The age processes at the follower nodes are statistically identical, and illustrated by $\Delta_{\text{follower}}(t)$ in Fig.~\ref{fig:ageEvol}, since the write time to the follower nodes are iid random variables. (However, in each update round, a random subset of followers receives the data at random times.)

    \begin{figure}[t]
        \centering
        \includegraphics[width=\columnwidth, trim={20 20 10 20},clip]{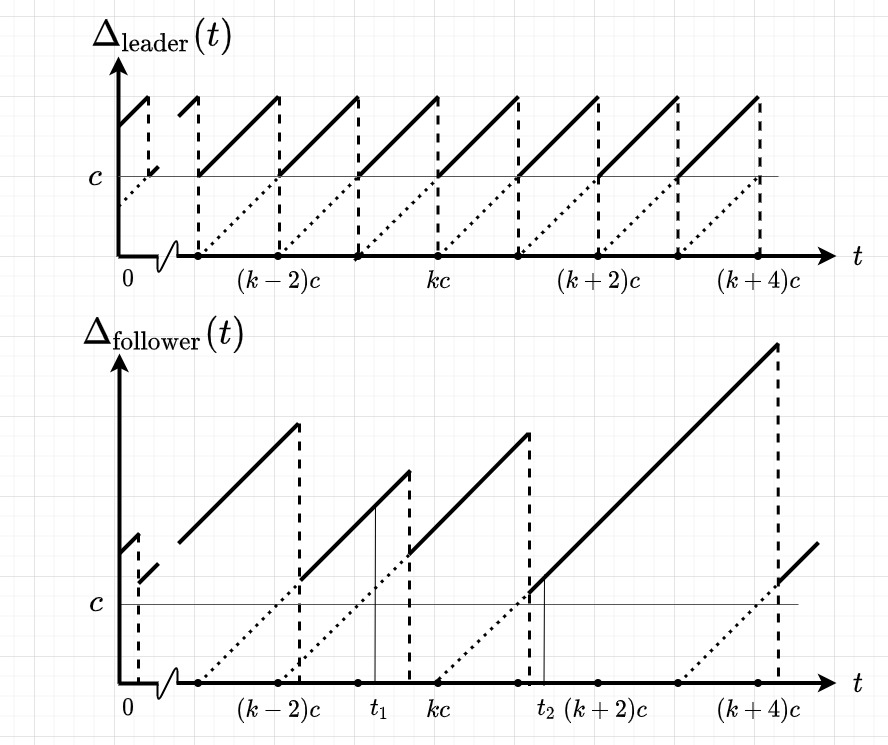}
        \caption{The age process for a sample leader node and a sample follower node. The leader nodes experience the same age process, illustrated by $\Delta_{\text{leader}}(t)$. The follower nodes may experience different yet statistically identical age process, a sample path of which is illustrated by $\Delta_{\text{follower}}(t)$.}
    \label{fig:ageEvol}
    \end{figure}

\section{Age Analysis}
For a read query, $r$ instances of the age processes are sampled from the collective set of $l$ leaders and $n-l$ followers. The age of the read is the minimum of the ages at the queried  nodes.  
The age of the retrieved data depends on the the types of nodes a read is sent to. Specifically, we analyse the average age of a read query based on the following partition:
    \begin{itemize}
        \item $B_1$: At least one node is chosen from the leaders,
        \item $B_2$: No node is chosen from the leaders.
    \end{itemize}
 The probability of having at least one leader queried is
    \begin{equation}
        \textup{Pr}\{B_1\}=\begin{cases}
              1-\frac{\binom{n-l}{r}}{\binom{n}{r}} & r\leq n-l,\\
              1    & r>n-l.
        \end{cases}
    \end{equation}
In what follows, we study the non-trivial case $r\leq n-l$.

Since leaders are guaranteed to receive every update, they are always consistent, and with the freshest data. Therefore, the age of the retrieved data in event $B_1$ is the age at leaders, characterized by the following lemma.
\begin{lemma}
Given $B_1$, the average age of a read query is
    \begin{equation}
        \mathbbm{E}[\Delta|B_1]=3c/2.
    \end{equation}
\end{lemma}
\begin{proof} 
When $B_1$ occurs, the instantaneous age of the read is $c+T_a$, where $c$ is the commit time of a write operation and $T_a$ is the arrival time of the read query, which is uniformly distributed in the interval $\left[0,c\right)$. Then the average age of a read can be calculated as
  \begin{align*}
            \mathbbm{E}[\Delta|B_1]&=\int_0^cf_{T_a}(t)\mathbbm{E}[\Delta|B_1,T_a=t]dt,\\
            &=\frac{1}{c}\int_0^c(c+t)dt=3c/2.\qquad\qquad\qquad\qquad
            \qedhere
\end{align*}
\end{proof}
When $B_2$ occurs, the age of the read query is the age at the most recently updated follower, among the $r$ followers the read query is sent to. A follower is consistent at the arrival time of a read $T_a$ if it is written with the most recent update by $T_a$. For instance, in 
Fig.~\ref{fig:ageEvol}, the follower would not return the most recent $(k-2)$th update to a read arriving at $T_a=t_1$. Whereas, it would return the most recent $k$th update to a read arriving at $T_a=t_2$. The following lemma gives the average age of a read query in event $B_2$.
\begin{lemma}
Given $B_2$, the average age of a read query is
    \begin{equation}
        \begin{split}
            &\mathbbm{E}[\Delta|B_2]=\frac{3c}{2}+\frac{\int_0^c\left[1-F_w(t)\right]^rdt}{1-[1-F_w(c)]^r}.
        \end{split}
    \label{equ:EAgeA2_2}
    \end{equation}
\end{lemma}
\begin{proof}
Under event $B_2$, the data is read from one of the $r$ selected followers, which we index $(1),(2),\dots,(r)$. For a given update to be successful at a follower, it should be written at most $c$ units of time past the commitment of that update. Thus, there is a fixed probability $p_c=F_w(c)$ (across the update rounds) that a follower successfully receives an update. Accordingly, we define the geometric random variable $Z_{(i)}$ as the number of update rounds follower $(i)$ has missed after its latest update. Thus, $Z_{(i)}\sim \text{Geo}(1-p_c)$. At the time of arrival $T_a$, the age at the selected follower $i$ is
    \begin{equation}
        \Delta_{(i)}=Z_{(i)}c+T_a.
    \end{equation}
The age of a read query is
    \begin{equation}
        \Delta=\min\{\Delta_{(1)},\Delta_{(2)},\dots,\Delta_{(r)}\}.
    \end{equation}
Therefore, the average age of a read can be written as
    \begin{equation}
        \begin{split}
            \mathbbm{E}[\Delta|B_{2}]&=\mathbbm{E}[\min\{Z_{(1)},\dots,Z_{(r)}\}|B_2]c+\mathbbm{E}[T_a|B_{2}],\\
            &=\mathbbm{E}[Z_{\min}|B_2]c+\mathbbm{E}[T_a],
        \end{split}
    \label{equ:EAgeA2}
    \end{equation}
where $Z_{\text{min}}$ is the number of missed update rounds at the most recently updated follower, among the $r$ followers the read query is sent to. Since the $Z_{(i)}$ are independent,
        \begin{align}
            \mathbbm{E}[Z_{\text{min}}|B_2]&=\int_{0}^{c}f_{T_a}(t)\mathbbm{E}[Z_{\text{min}}|T_a=t,B_2]dt\nn
            &=\frac{1}{c}\int_{0}^{c}\sum_{z=0}^\infty\left(\textup{Pr}\{Z_{(i)}>z|T_a=t,B_2\}\right)^rdt
            \label{equ:EZmin}
        \end{align}
Defining $F_w(t)=p_t$,
        \begin{align}
            \textup{Pr}\{Z_{(i)}=1|T_a=t,B_2\}&=p_t,\nn
            \textup{Pr}\{Z_{(i)}=2|T_a=t,B_2\}&=(1-p_t)p_c,\nn
            \textup{Pr}\{Z_{(i)}=3|T_a=t,B_2\}&=(1-p_t)(1-p_c)p_c,\nn
            &\vdots\nn
            \textup{Pr}\{Z_{(i)}=j|T_a=t,B_2\}&=(1-p_t)(1-p_c)^{j-2}p_c.
        \end{align}
Therefore,
        \begin{align}
            \textup{Pr}\{Z_{(i)}>0|T_a=t,B_2\}&=1,\nn
            \textup{Pr}\{Z_{(i)}>1|T_a=t,B_2\}&=1-p_t,\nn
            \textup{Pr}\{Z_{(i)}>2|T_a=t,B_2\}&=(1-p_t)(1-p_c),\nn
            &\vdots\nn
            \textup{Pr}\{Z_{(i)}>z|T_a=t,B_2\}&=(1-p_t)(1-p_c)^{z-1}.\label{equ:CCDF}
        \end{align}
Substituting (\ref{equ:CCDF}) in (\ref{equ:EZmin}) gives
        \begin{align}
            \mathbbm{E}[Z_{min}|B_2]&=\frac{1}{c}\int_{0}^{c}\sum_{z=0}^\infty\left(\textup{Pr}\{Z_{(i)}>z|T_a=t,B_2\}\right)^rdt\nn
            &=\frac{1}{c}\int_{0}^{c}\left[1+\sum_{z=1}^\infty(1-p_t)^r(1-p_c)^{r(z-1)}\right]dt\nn
            &=1+\frac{1}{c}\sum_{z=1}^\infty(1-p_c)^{r(z-1)}\int_0^c\left[1-F_w(t)\right]^rdt\nn
            &=1+\frac{1}{c}\frac{\int_0^c\left[1-F_w(t)\right]^rdt}{1-(1-p_c)^r}.\label{equ:EZmin2}
        \end{align}
    
Finally, by substituting (\ref{equ:EZmin2}) in (\ref{equ:EAgeA2}), 
    \begin{equation*}
        \begin{split}
            \mathbbm{E}[\Delta|B_{2}]&=\frac{3c}{2}+\frac{\int_0^c\left[1-F_w(t)\right]^rdt}{1-[1-F_w(c)]^r}.
        \end{split}
 \qquad\qquad\qquad       \qedhere
    \end{equation*}
\end{proof}

    \begin{figure}[t]
        \centering
        \includegraphics[width=\columnwidth,trim={8 3 40 40},clip]{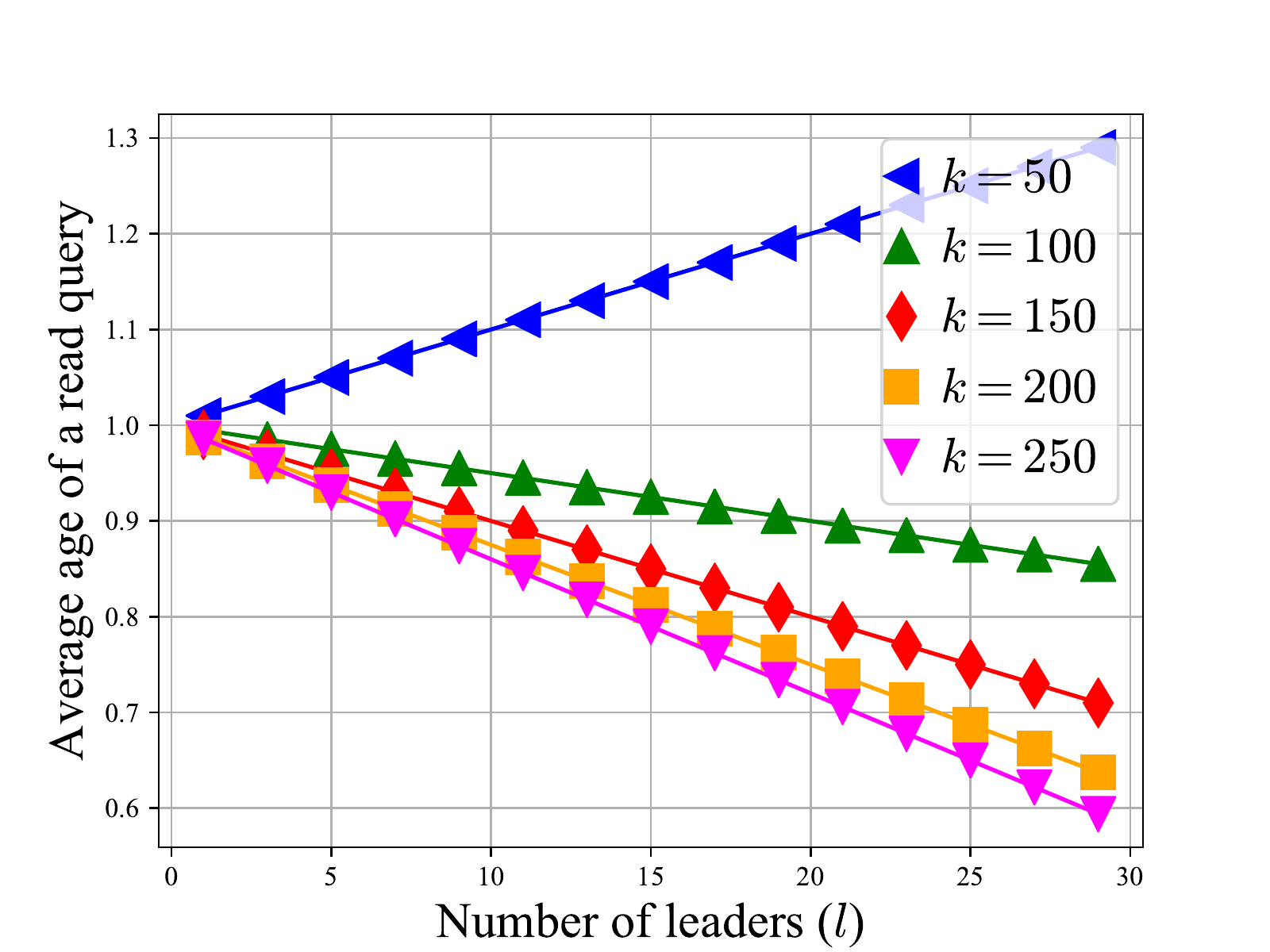}
        \caption{The average age of a read query vs. the number of leaders, see (\ref{equ:r1}), when $n=50$, $r=1$ and $\lambda=1$. Increasing $l$ either monotonically increases the average age ($k=50$) or monotonically increases it ($k\geq100$).}
    \label{fig:lefkt_r1}
    \end{figure}

\begin{theorem}
When the follower writing time $T_w$ is exponential$(\lambda)$, the average age of a read query is
    \begin{equation}
        \mathbbm{E}[\Delta]=\frac{3c}{2}+\frac{\binom{n-l}{r}}{\binom{n}{r}}\frac{1}{\lambda r}.
    \label{equ:avgAgeExp}
    \end{equation}
\end{theorem}
\begin{proof}
With $F_w(t)=1-e^{-\lambda t}$,
    \begin{equation*}
        \int_0^c\left[1-F_w(t)\right]^rdt=\frac{1}{\lambda r}\left(1-e^{-\lambda rc}\right),
    \end{equation*}
and, $1-[1-F_w(c)]^r=\left(1-e^{-\lambda rc}\right)$, which by substituting in (\ref{equ:EAgeA2_2}) completes the proof.
\end{proof}

\section{Numerical Results}
In this section, we present numerical results based on our average age analysis. First, we study the effect of the number of leaders on the average age of a read query. Since the writes to the leaders occur sequentially, the number of leaders should affect the commit time $c$. On the other hand, knowing that all the nodes are the members of the same ecosystem, there should be a connection between the writing time to a leader and that of a follower node. Accordingly, we assume $c$ is related to the number of leaders $l$ and the rate $\lambda$ of the write process to the followers by
    \begin{equation}
        c=\frac{l}{k\lambda},
    \end{equation}
where $k$ can be interpreted as the relative speed of writing to the leaders compared to writing to a follower. Note that this is not the only model to scale $c$ with the other system parameters; there are other possibilities. We think of $k$ as an inherent factor of the underlying system and show our results for a range of values. With this model and $T_w\sim \text{Exp}(\lambda)$, the average age of a read query is
    \begin{equation}
        \mathbbm{E}[\Delta]=\frac{1}{\lambda}\left[\frac{3l}{2k}+\frac{1}{r}\frac{\binom{n-l}{r}}{\binom{n}{r}}\right].
    \label{equ:avgAgeModel}
    \end{equation}
We further assume that the time unit is chosen such that $\lambda=1$ and the average write time to a follower is one time unit.
    \begin{figure}[t]
        \centering
        \includegraphics[width=\columnwidth,trim={8 3 40 40},clip]{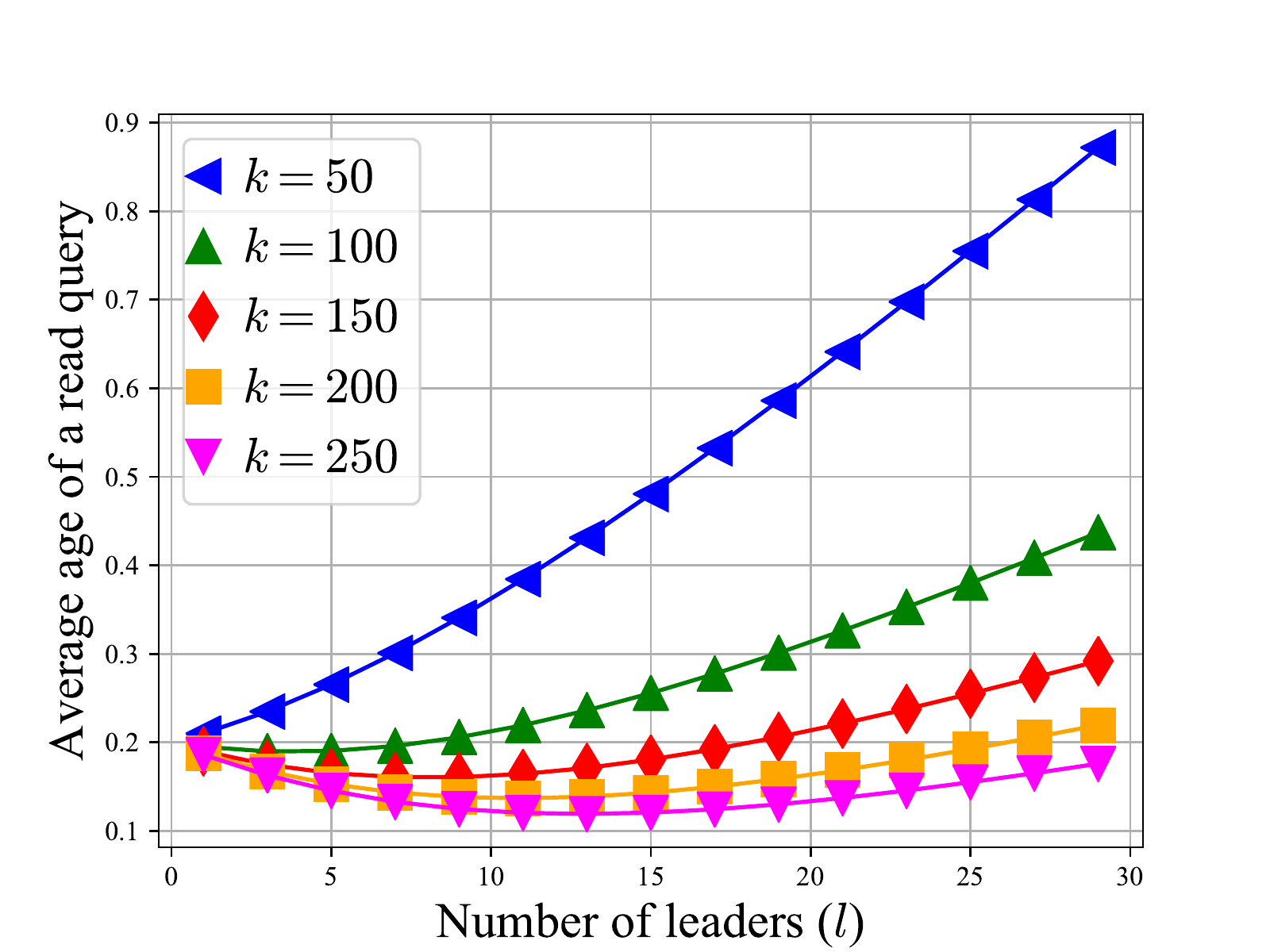}
        \caption{The average age of a read query vs. the number of leaders, see (\ref{equ:avgAgeExp}), when $n=50$, $r=4$ and $\lambda=1$. With small $k$, the average age increases monotonically with $l$. With larger $k$, it decreases initially, reaches an optimum point and then it increases.}
    \label{fig:lefkt_r5}
    \end{figure}

Fig. \ref{fig:lefkt_r1} shows the variation of the average age of a read query as the number of leaders change, when $n=50$ and $r=1$. The effect of increasing the number of leaders on the average age depends on the parameter $k$. With $k=50$ the average age increases monotonically; while for $k\geq100$, it decreases, also monotonically, with $l$. Intuitively, when writing on leaders is relatively slow, i.e. $k$ is small, having more leaders within a fixed number of nodes increases the staleness of an update, since it is available to read only after it is committed. On the other hand, when writing to the leaders is fast ($k$ is large), then having more leaders increases the chance of a read query to be sent to a consistent node. This also can be verified from (\ref{equ:avgAgeModel}) by substituting $\lambda=1$, $r=1$ and $n=50$,
    \begin{equation}
        \mathbbm{E}[\Delta]=1+\left[\frac{75}{k}-1\right]\frac{l}{n}.
        \label{equ:r1}
    \end{equation}

In Fig. \ref{fig:lefkt_r5}, the effect of number of leaders on average age is illustrated when $n=50$ and $r=5$. With a larger size of read query, the probability that a read query is sent to a consistent node is higher. For that reason the average age is dropped in Fig. \ref{fig:lefkt_r5} compared to Fig. \ref{fig:lefkt_r1}. On the other hand, the monotonic changes in the average age do not occur for every value of $k$ with $r=5$. For $k\geq100$, the average age decreases initially with the number of leaders and after reaching an optimum point it starts increasing. The reason for this behaviour can be explained as follows. Increasing the number of leaders has two competing effects on the average age. First, by increasing the commit time it increases the staleness of an update. Second, by increasing the number of leaders and, at the same time, reducing the number of followers it increases the probability of a read query to be sent to a consistent and receive (relatively) less stale data. Now, when $r$ is large, the probability that a read query is sent to a consistent node is already high. Therefore, the second effect has smaller contribution to the average age. Thus, increasing the number of leaders, even when they are fast to write, does not decrease the average age monotonically. From (\ref{equ:avgAgeModel}), the range of $k$ for which the average age initially decreases is $k\geq\lceil{3n(n-1)/2(n-r)}\rceil$. 
With $n=50$ and $r=5$, the average age initially decreases with $k\geq82$.
    \begin{figure}[t]
        \centering
        \includegraphics[width=\columnwidth,trim={8 3 40 40},clip]{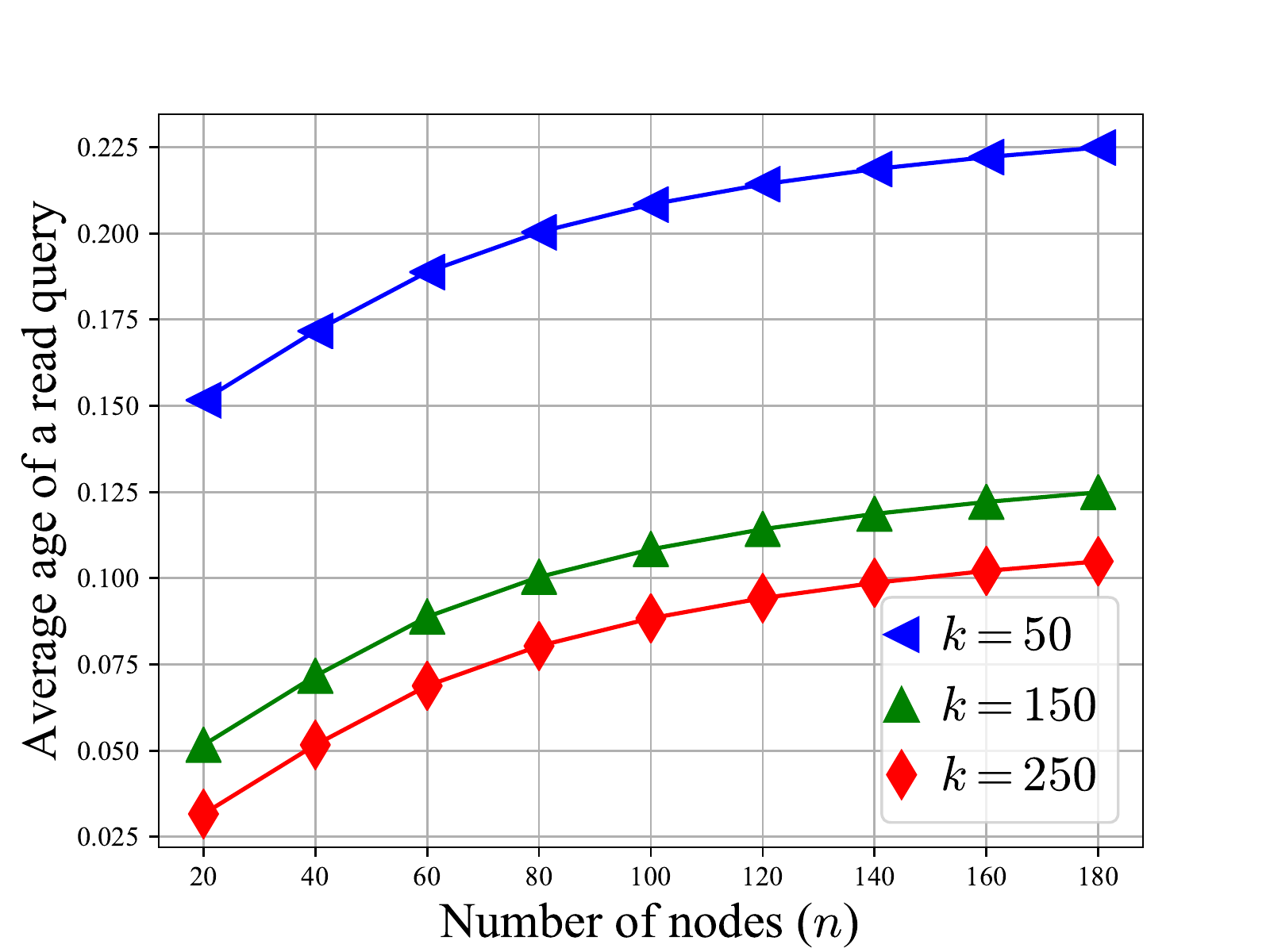}
        \caption{Average age of a read query vs. the number of nodes, see (\ref{equ:avgAgeExp}), when $r=10$, $l=5$ and $\lambda=1$. The average age increases monotonically with the number of nodes.}
    \label{fig:nefkt_rfix}
    \end{figure}

It is clear that increasing the size of a read query $r$ reduces the average age of its retrieved data. Furthermore, with fixed $l$, increasing the number of nodes reduces the probability of a read query to be sent to a consistent node and thus it increases the average age, as it is shown in Fig. \ref{fig:nefkt_rfix}. Nevertheless, the effect of changing both $n$ and $r$, at the same time, is more complex. Consider the following scenario. A leader-based database has stored a data chunk, with parameters $n$, $r$ and $l$. Due to a growth in the demand, the system administrator decides to increase the number of replicas $n$, but without deteriorating the performance. To that end, she could increase the size of read queries $r$ and therefore may ask the following question: how much increment in $r$ is needed when $n$ increases to maintain the same average age? Our numerical results show that, with establishing a relationship between $n$ and $r$, one can possibly avoid the increase in the average age as $n$ increases. For instance, in Fig. \ref{fig:nefkt_rvar}, the relationship between the two parameters is $r=10+n/20$. With this dependence, we can see that increasing the number of nodes beyond $n=100$ does not increase the average age. This result could be especially useful for systems with time-varying demands. Finding the exact dependence, however, is out of the scope of this paper and is left as an open problem for future study.
    \begin{figure}[t]
        \centering
        \includegraphics[width=\columnwidth,trim={8 3 40 40},clip]{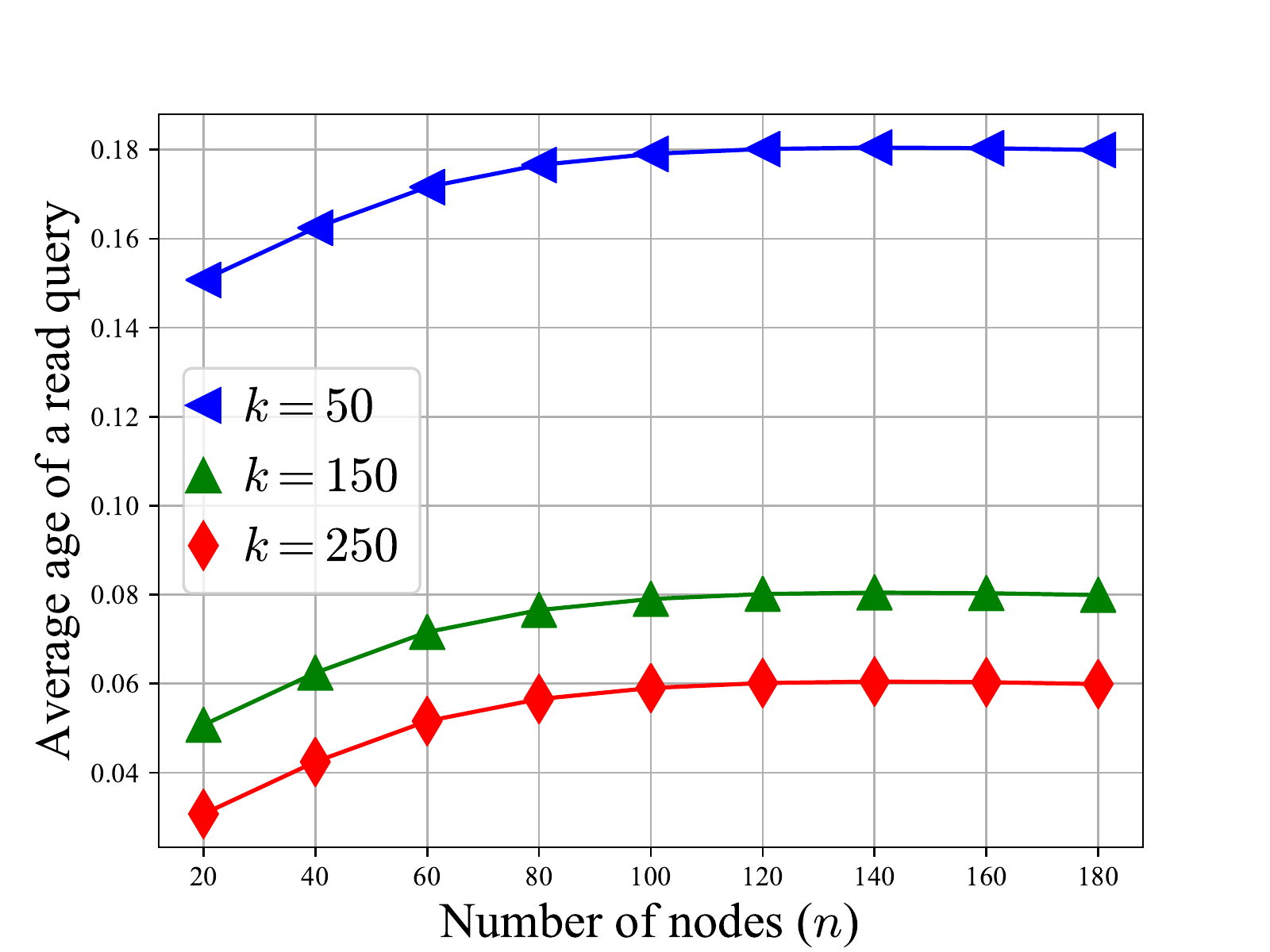}
        \caption{Average age of a read query vs. the number of nodes, see (\ref{equ:avgAgeExp}), when $l=5$ and $\lambda=1$. The size of the read query is dependent to the number of nodes, $r=10+n/20$. Beyond $n=100$ the average age does not increase with $n$.}
    \label{fig:nefkt_rvar}
    \end{figure}
\section{Conclusion}
The average age of the retrieved data by an instantaneous read query in a leader-based database was studied. With a deterministic model for the writing time on the leaders and a probabilistic model for the writing time on a follower, the average age of read query was analyzed. Closed-form average age under exponential distribution of the follower writing time was derived. It was shown numerically that, when the writing time on the leaders group scales linearly with the number of leaders, the average age could monotonically increase, monotonically decrease or behave non-monotonically, depending on the size of the read query and the relative speed of writing on a leader to writing on a follower. Furthermore, for a data with dynamic demands, it was shown that it is possible to keep the average age fixed when the number of follower nodes increases, with a slight increase in the size of the read query.

\section*{Acknowledgement}
This research was supported in part by the NSF awards No. CNS-1717041, CIF-1717314 and CCF-1559855.
\newpage
\bibliographystyle{IEEEtran}
\bibliography{ref,AOI-2019-01}
\end{document}